\documentclass[journal]{IEEEtran}
\usepackage{amsfonts}
\usepackage{amssymb}
\usepackage{graphicx}
\usepackage{subfigure}
\usepackage{amsmath}
\usepackage{color}
\usepackage{algorithm}
\usepackage{algpseudocode}
\usepackage{amsthm}
\usepackage[bookmarks=false, draft]{hyperref}
\usepackage{breakurl}

\usepackage{cite}

\newcommand{\bp}{\begin{proof} \small }
\newcommand{\ep}{\end{proof} \normalsize}
\newcommand{\epx}{\end{proof} \small}
\newcommand{\bpa}{\begin{proofappx} \footnotesize }
\newcommand{\epa}{\end{proofappx} \small }
\newtheorem{theorem}{Theorem}
\newtheorem{proposition}{Proposition}

\newtheorem{remark}{Remark}
\newtheorem*{theorem*}{Theorem}
\newtheorem*{proposition*}{Proposition}
\newtheorem*{corollary*}{Corollary}
\newtheorem*{lemma*}{Lemma}
\newtheorem*{assumption*}{Assumption}
\newtheorem*{definition*}{Definition}
\newtheorem*{claim*}{Claim}

\newcommand{\be}{\begin{equation}}
\newcommand{\ee}{\end{equation}}
\newcommand{\bs}{\begin{subequations}}
\newcommand{\es}{\end{subequations}}
\newcommand{\bq}{\begin{eqnarray}}
\newcommand{\eq}{\end{eqnarray}}
\newcommand{\bqn}{\begin{eqnarray*}}
\newcommand{\eqn}{\end{eqnarray*}}

\newcommand{\ba}{\left[ \begin{array}}
\newcommand{\ea}{\\ \end{array} \right]}
\newcommand{\ben}{\begin{enumerate}}
\newcommand{\een}{\end{enumerate}}

\def\a{{\boldsymbol{a}}}

\def\real{{\mathchoice%
{\hbox{\rm\setbox1=\hbox{I}\copy1\kern-.45\wd1 R}}
{\hbox{\rm\setbox1=\hbox{I}\copy1\kern-.45\wd1 R}}
{\hbox{\scriptsize\rm\setbox1=\hbox{I}\copy1\kern-.45\wd1 R}}
{\hbox{\scriptsize\rm\setbox1=\hbox{I}\copy1\kern-.45\wd1 R}}}}

\def\Zint{{\mathchoice{\setbox1=\hbox{\sf Z}\copy1\kern-.75\wd1\box1}
{\setbox1=\hbox{\sf Z}\copy1\kern-.75\wd1\box1}
{\setbox1=\hbox{\scriptsize\sf Z}\copy1\kern-.75\wd1\box1}
{\setbox1=\hbox{\scriptsize\sf Z}\copy1\kern-.75\wd1\box1}}}
\newcommand{\complex}{ \hbox{\rm C\kern-0.45em\rule[.07em]{.02em}{.58em}%
\kern 0.43em}}

\begin{document}
	
\title{EMM: Energy-Aware Mobility Management for Mobile Edge Computing in Ultra Dense Networks}
\author{Yuxuan~Sun,
	Sheng~Zhou,~\IEEEmembership{Member,~IEEE,}
	and Jie~Xu,~\IEEEmembership{Member,~IEEE}
	\thanks{Y. Sun and S. Zhou are with the Department of Electronic Engineering, Tsinghua University, China. Email: sunyx15@mails.tsinghua.edu.cn, sheng.zhou@tsinghua.edu.cn.  (Corresponding author: S. Zhou)}
	\thanks{J. Xu is with the Department of Electrical and Computer Engineering, University of Miami, USA. Email: jiexu@miami.edu.}
	\thanks{This work is sponsored in part by the Nature Science Foundation of China No. 61571265, No. 91638204, No. 61621091, and Intel Collaborative Research Institute for Mobile Networking and Computing.}
	\thanks{Part of this work has been published in IEEE ICC 2017 \cite{xu2017e2m2}.}
}

\maketitle

\begin{abstract}
	Merging mobile edge computing (MEC) functionality with the dense deployment of base stations (BSs) provides enormous benefits such as a real proximity, low latency access to computing resources. However, the envisioned integration creates many new challenges, among which mobility management (MM) is a critical one. Simply applying existing radio access oriented MM schemes leads to poor performance mainly due to the co-provisioning of radio access and computing services of the MEC-enabled BSs. In this paper, we develop a novel user-centric energy-aware mobility management (EMM) scheme, in order to optimize the delay due to both radio access and computation, under the long-term energy consumption constraint of the user. Based on Lyapunov optimization and multi-armed bandit theories, EMM works in an online fashion without future system state information, and effectively handles the imperfect system state information. Theoretical analysis explicitly takes radio handover and computation migration cost into consideration and proves a bounded deviation on both the delay performance and energy consumption compared to the oracle solution with exact and complete future system information. The proposed algorithm also effectively handles the scenario in which candidate BSs randomly switch on/off during the offloading process of a task. Simulations show that the proposed algorithms can achieve close-to-optimal delay performance while satisfying the user energy consumption constraint.
\end{abstract}


\section{Introduction}
Ultra dense networking (UDN) \cite{quek2013small} and mobile edge computing (MEC) (a.k.a. fog computing) \cite{hu2015mobile} \cite{mao2017mobile} are regarded as key building blocks for the next generation mobile network. UDN increases the network capacity through the ultra-dense deployment of small cell base stations (BSs), as a key technology addressing the so-called 1000x capacity challenge \cite{chen2014requirements}. MEC provides cloud computing and storage resources at the edge of the mobile network, creating significant benefits such as ultra-low latency, intensive computation capabilities while reducing the network congestion, which are necessary for emerging applications such as Internet of things, video stream analysis, augmented reality and connected cars\cite{etsimec004}. 

It is envisioned that endowing each radio access node with cloud functionalities will be a major form of MEC deployment scenarios, i.e., MEC-enabled UDN\cite{etsimec002}. However, current studies on UDN and MEC are mostly separate efforts. Despite the enormous potential benefits brought by the integration of UDN and MEC, a key challenge for the overall system performance is mobility management (MM),
which is the fundamental function of associating mobile devices with appropriate BSs on the go, thereby enabling mobile services (i.e. radio access and computing) to be delivered. Traditionally, MM was designed for providing radio access only. Merging UDN and MEC drastically complicates the problem. Simply applying existing solutions leads to poor MM performance mainly due to the co-provisioning of radio access and computing services. In particular, MM for MEC in UDN faces the following three major challenges:

1) The first challenge is the lack of accurate information (e.g., radio access load, computation load, etc.) of candidate BSs on the user side, especially when MM is carried out in a user-centric manner. If the user does not know a priori which BS offers the best performance, the MM can be very difficult.

2) An even severe challenge is the unavailability of future information (e.g., future tasks for computation offloading, candidate BSs, channel conditions, available edge cloud resources, etc.). Since the mobile user has limited battery power, the long-term energy budget couples the short-term MM decisions across time, and yet the decisions have to be made without foreseeing the future.

3) Moreover, UDN is a very complex and volatile network environment due to the fact that many small cell BSs are owned, deployed and managed by end-users. In addition, the operator often implements BS sleeping techniques for energy saving. As a result, candidate BSs can be randomly switched on/off over time, thus demanding for a MM algorithm that can fast track the optimal BS for performance optimization. 


\subsection{Related Work}
Mobile edge computing has received an increasing amount of attentions recently, see \cite{mao2017mobile} for a comprehensive survey. A central theme of many prior studies is to design task offloading policies and resource management schemes, i.e. what/when/how to offload a user's workload from its device to the edge system or cloud, and how much radio and computing resources should be allocated to each user. For a single-user MEC system, an energy-optimal binary offloading policy is proposed in \cite{zhang2013energy} by comparing the energy consumption of local execution and offloading, while a delay-optimal task scheduling policy with random task arrivals is proposed in \cite{liu2016delay}. For multi-user MEC systems, both centralized \cite{you2016energy} and distributed \cite{chen2016efficient} radio and computation resource management schemes are studied to optimize system-level performance. However, most of the existing works consider a single MEC server, and overlook the user mobility issue.   

Mobility management has been extensively investigated in LTE systems. For example, the solutions in \cite{xenakis2014mobility} work efficiently in less-densified heterogeneous networks, but may bring new problems such as frequent handover and the Ping-Pong effect when the network density becomes high \cite{lopez2012mobility}. To address this challenge, an energy-efficient user association and power control policy is proposed in \cite{7842367}, while a learning-based MM scheme is proposed in \cite{shen2016non} based on the multi-armed bandits (MAB) theory \cite{auer2002finite}. Both schemes work in a user-centric manner, which has been an emerging trend of MM for the future 5G network \cite{chen2016user}. However, all these works merely consider the radio access. Endowing BSs with MEC capabilities requires new MM solutions. 

There are a few works considering service migration, which is a key component of MM in MEC. 
An optimal computation migration policy is designed in \cite{taleb2016follow}, in order to reduce the migration cost while maintaining good user quality of service. The optimal policy is proved to be threshold-based w.r.t. the migration cost and backhaul data transmission cost in \cite{wang2014mobility}. However, the radio access aspect has not been considered in these works.

Motivated by the limitations of the current literature, we design user-centric MM algorithms in MEC-enabled UDN in this paper. 
Our work aims to provide guidance to the user about which BS and MEC server should be selected and when to perform handover, with the challenges of lacking both the accurate future information and current BS-side information. By integrating the Lyapunov optimization technique \cite{neely2010stochastic} and MAB theory \cite{auer2002finite}, we solve an average delay minimization problem under a long-term energy budget constraint, and prove that our proposed algorithms can provide strong performance guarantee. 
Note that our work provides the BS association decisions, which can be supported by the link layer handover protocols \cite{lopez2012mobility}, while further served as the basis of the network layer MM protocols, such as Proxy Mobile IPv6 protocol\cite{modares2016asurvey},
Different from the conference version of this work \cite{xu2017e2m2}, we introduce a more general model considering transmission delay and BS handover cost, and provide new theoretical analysis and simulation results. Moreover, we develop a new algorithm based on the volatile MAB (VMAB) framework \cite{bnaya2013social} to handle random BS on/off during task offloading. 



\subsection{Contributions}

1) We develop a novel energy-aware user-centric MM scheme, called EMM, to overcome the aforementioned challenges by leveraging the combined power of Lyapunov optimization and MAB theories. The proposed EMM algorithm can deal with various practical deployment scenarios, including those in which the user has limited BS-side information and the BSs dynamically switch on and off. 

2) We rigorously characterize the performance of the proposed EMM algorithms. We prove that the EMM algorithms can achieve close-to-optimal performance within a bounded deviation without requiring future system information, while satisfying the long-term energy budget constraint. Moreover, we quantify the performance loss due to learning the BS-side information in terms of the learning regret, explicitly taking into account the additional cost caused by radio handover, computation migration and varying candidate BSs. 

3) Extensive simulations are carried out to evaluate the performance of the EMM algorithm and validate our theoretic findings. The results confirm that our proposed algorithm can achieve close-to-optimal delay performance compared to the oracle solution with exact and complete future system information, while satisfying the energy consumption constraint of the user. Simulations also reveal the impact of design parameters on the system performance, thereby providing guidelines for real-world deployment of MEC in UDN.

The rest of this paper is organized as follows. 
We describe the system model and formulate the problem in Section~\ref{secmodel}.  Section~\ref{secgsi} and \ref{seclsi} develop EMM algorithms and conduct performance analysis. Section \ref{seclsid} extends the algorithm to handle varying BS sets. Simulation results are provided in Section \ref{secsim}, followed by the conclusion in Section \ref{seccon}.


\section{System Model and Problem Formulation} \label{secmodel}
\begin{figure}
	\centering
	\includegraphics[width=0.48\textwidth]{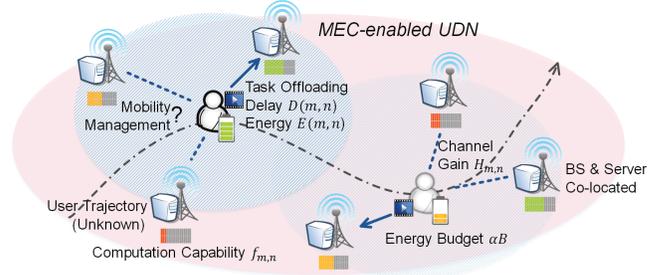}
	\hspace{-0.2in}
	\caption{Illustration of the considered user-centric MM in MEC-enabled UDN. A representative user with unknown trajectory offloads each computation task $m$ to one of the candidate BSs $n$, based on the overall delay $D(m,n)$ (the sum of communication, computation and handover delay) and the energy consumption $E(m,n)$ for data transmission. The objective is to minimize the average delay under the energy consumption budget $\alpha B$.}\label{system}
	\vspace{-0.2in}
\end{figure}

\subsection{Network Model}
We consider a network with $N$ densely deployed BSs indexed by $\mathcal{N}=\{1,2,...,N\}$,  as shown in Fig. \ref{system}.
Each BS is endowed with cloud computing functionalities, which is considered as one of the main deployment scenarios of MEC \cite{etsimec002}. We focus on a representative mobile user moving in the network, who generates totally $M$ computation tasks over time, and these tasks are offloaded to the BS for computing. Let $L_m$ denote the location where task $m$ is generated. No prior knowledge about the user trajectory is required. In other words, our work is applicable to any mobility model, such as the random waypoint model or others as described in \cite{batabyal2015mobility}.

Multiple BSs can provide service to the user at any location $L_m$ due to the dense deployment. Denote $\mathcal{A}(L_m) \subseteq \mathcal{N}$ as the set of BSs that cover location $L_m$. After task $m$ is generated, the MM scheme makes decisions on which BS serves the user, among the set of BSs $\mathcal{A}(L_m)$. 
We design user-centric MM schemes, i.e., the user makes the BS association and handover decisions, which is a promising candidate in the 5G standards \cite{chen2016user}. Moreover, we focus on a local computation scenario, i.e., the associated BS is responsible for providing both radio access and edge computing services without further offloading the computation tasks to other BSs or the remote cloud.


\subsection{Computation Task and Service Model}
A widely used three-parameter model (see \cite{mao2017mobile} and references therein) is adopted to describe each computation task $m$: input data size $\lambda_m\in [0, \lambda_\text{max}]$ (in bits) that needs to be offloaded, computation intensity $\gamma_m \in [0, \gamma_\text{max}]$ (in CPU cycles per bit) indicating how many CPU cycles are required to compute one bit input data, and completion deadline $D_m$. Parameters $\lambda_\text{max}$ and $\gamma_\text{max}$ are the maximum possible input data size and computation intensity, respectively. 



Each computation task is relatively large and hence can be further divided into many subtasks that must be processed in sequence (e.g. computing the subsequent subtasks requires the results of the previous subtasks). 
Taking video stream analytics as an example, like object detection or tracking from a video stream, the analysis can be operated on the edge server using the Hadoop MapReduce framework \cite{anjum2016video}. A relatively long video frame is further divided into many short video clips through video segmentation, each having a number of video frames. Note that our work is orthogonal to the video segmentation problem \cite{5539893}, and we omit the overhead of video segmentation for simplicity, which can be seen as an additional constant delay to the system performance.
Let $K_m \leq \bar{K}$ be the number of subtasks of task $m$, where $\bar{K}$ is maximum number of subtasks. Assume that subtasks are of the equal size $\lambda_0$ for analytical simplicity (hence $\lambda_m = K_m\lambda_0$). Nevertheless, our framework can handle subtasks of heterogeneous sizes.

Each BS $n\in\mathcal{N}$ is equipped with an MEC server of maximum CPU frequency $F_n$ (in CPU cycles per second), and can provide computation services for multiple tasks from multiple users simultaneously using processor sharing. We use computation capability $f_{m,n}$ to describe the CPU frequency that BS $n$ can allocate to task $m$, which depends on several factors on the BS side, such as the maximum CPU frequency $F_n$, the current total workload intensity, etc. We assume that $f_{m,n}$ does not change during the processing of one task but can change across tasks. If BS $n$ is selected to compute a subtask of size $\lambda_0$ and computation intensity $\gamma_m$, then given the allocated CPU frequency $f_{m,n}$, the computation delay is
\begin{align}
d_c(m,n) = \frac{\lambda_0\gamma_m}{f_{m,n}}.
\end{align}

\subsection{Communication and Energy Consumption Model}
The input data is transmitted from the user to the serving BS through the wireless uplink channel. Denote $H_{m,n}$ as the channel gain between the user at location $L_m$ and BS $n\in \mathcal{A}(L_m)$. We assume that during the computation of each task $m$, the user does not move much and hence $H_{m,n}$ is constant. Nevertheless, if the user moves considerably, we consider that one task is divided into multiple subtasks, and for each subtask the user stays more or less at the same location. Given the transmission power $P_{\mathrm{tx}}$ of the user, the maximum achievable uplink transmission rate is given by:
\begin{align}
r(m,n) = W\log_2\left(1 + \frac{P_{\mathrm{tx}}H_{m,n}}{\sigma^2 + I_{m,n}}\right),
\end{align}
where $W$ is the channel bandwidth, $\sigma^2$ is the noise power and $I_{m,n}$ is the inter-cell interference power at BS $n$ while offloading task $m$. The transmission delay for sending the input data of size $\lambda_0$ to BS $n$ is thus
\begin{align}  \label{transdelay}
d_t(m,n) = \frac{\lambda_0}{r(m,n)}.
\end{align}
Also, the energy consumption for offloading a subtask for task $m$ is therefore
\begin{align}
e(m, n) = \frac{P_{\mathrm{tx}} \lambda_0}{r(m,n)}.
\end{align}

\begin{remark}
	Downlink transmission delay and packet loss are not considered in this work. Nevertheless, the following analysis and the proposed solutions are still applicable with these considerations. For example, downlink transmission delay and packet loss can be reflected by additional transmission delay that changes expression \eqref{transdelay}.
\end{remark}

\subsection{Handover and Migration Cost Model}
For each computation task $m$, its subtasks must be computed in sequence, but can be offloaded to different BSs. 
This may be because the user learns that the serving BS's computing capability is weak (we will introduce the learning problem in Section \ref{seclsi}) and hence decides to switch to a different BS in its vicinity or BSs can appear or disappear in the transmission range of the user due to dynamic BS on/off for energy saving \cite{7114343}. When consecutive subtasks are processed on different BSs, an additional delay cost is incurred due to the handover procedure and the computation migration. Let $C_m$ be the one-time handover cost for task $m$. Given the sequences of BSs that serve its subtasks, denoted by $\a_m = (a_m^1,a_m^2,...,a_m^{K_m})$, the overall handover cost for task $m$ is
\begin{align}
h(m, \a_m) = C_m \sum_{k=2}^{K_m}\mathbb{I}\{a^k_m\neq a^{k-1}_{m}\},
\end{align}
where $a^k_m \in \mathcal{A}(L_m)$ is the serving BS for subtask $k$ of task $m$, and $\mathbb{I}\{x\}$ is an indicator function with $\mathbb{I}\{x\} = 1$ if event $x$ is true and $\mathbb{I}\{x\} = 0$ otherwise.

\subsection{Problem Formulation} \label{secfor}
Mobile users often have limited energy budgets (e.g., due to limited battery capacity). Therefore, the objective of the mobile user is to make MM decisions, specifically which BS to associate and when to perform handover, in order to minimize the average delay given its limited energy budget. For task $m$, the overall delay is
\begin{align}
D(m, \a_m) = \sum_{k=1}^{K_m} d(m, a_m^k) + h(m, \a_m),
\end{align}
where $d(m, a_m^k) \triangleq d_c(m, a_m^k) + d_t(m, a_m^k)$ is the sum of computation delay and uplink transmission delay for subtask $k$. The overall energy consumption for processing task $m$ is
\begin{align}
E(m, \a_m) = \sum_{k=1}^{K_m} e(m, a_m^k).
\end{align}

Formally, the problem is formulated as follows
\begin{align}
\textbf{P1:}&~~\min_{\a_1,...,\a_M} ~\frac{1}{M}\sum_{m=1}^{M} D(m, \a_m) \label{obj}\\
\text{s.t.} &~~\sum_{m=1}^{M} E(m, \a_m)\leq \alpha B \label{budget}\\
&~~D(m, \a_m) \leq D_m, ~\forall m \label{maxdelay}\\
&~~a^k_m \in \mathcal{A}(L_m), \forall m, \forall k = 1,2,...,K_m \label{coverage}.
\end{align}
The first constraint \eqref{budget} states that the total energy consumption is limited by the energy budget of the user, where $\alpha \in (0, 1]$ indicates the desired capping of energy consumption relative to the total battery capacity $B$. 
The second constraint \eqref{maxdelay} requires that the overall delay for processing task $m$ does not exceed the completion deadline $D_m$. Note that even if we set up a deadline for each task, the user still prefers to receive the result as soon as possible.
The last constraint \eqref{coverage} states that the associated BSs are those that cover location $L_m$.

There are two major challenges to solve problem \textbf{P1}. First, optimally solving \textbf{P1} requires complete non-causal information over the entire trip of the user, including parameters of all tasks, user trajectory, traffic intensity of all BSs, etc., which is impossible to acquire in advance. Furthermore, \textbf{P1} belongs to integer nonlinear programming problem. Even if the complete future information is known a priori, it is still difficult to solve due to the high complexity. Therefore, we will propose online algorithms that can efficiently make MM decisions without the future information.

\subsection{Oracle Benchmark and Theoretical Upper Bound} \label{jstep}
In this subsection, we describe an algorithm that knows the complete future information for the next $J$ computation tasks. Albeit impractical, the purpose of introducing this algorithm is merely to provide theoretical upper bounds on the performance of any practical online algorithm. We will prove later that our proposed algorithm achieves close-to-optimal performance by comparing to this oracle benchmark.

The $J$-step lookahead problem is defined as 
\begin{align}
\textbf{P2:}&~~\min_{\a_{rJ+1},...,\a_{(r+1)J}}~ \frac{1}{J}\sum_{m=rJ+1}^{(r+1)J }D(m, \a_m)\\
\text{s.t.} &~~\sum_{m=rJ+1}^{(r+1)J} E(m, \a_m) \leq \frac{\alpha B}{R} \label{budgetJ}\\
&~~\text{constraints  \eqref{maxdelay}, \eqref{coverage}} \label{constraintJ}.
\end{align}
The entire trip of the user is divided into $R \geq 1$ frames. In each frame, the user generates $J \geq 1$ tasks and hence $M = RJ$. We assume that there is an oracle that provides accurate information of the subsequent $J$ tasks at the beginning of each frame. Given this information, the user can obtain the MM decisions for the next $J$ tasks by solving the $J$-step lookahead problem \textbf{P2}.

Clearly if $R = 1$, then the $J$-step lookahead problem is the original offline problem \textbf{P1}. 
Assume that for all $r=0,1,..., R-1$, there exists at least one sequence of MM decisions $\a_{rJ+1}, ..., \a_{(r+1)J}$ that satisfy the constraints of \textbf{P2}. 
Denote $g^*_r$ as the optimal average delay achieved by \textbf{P2} in the $r$-th frame. Thus $g^*=\frac{1}{R}\sum_{r=0}^{R-1}g^*_r$ is the minimum long-term average delay achieved by the $J$-step lookahead problem.


\section{Online Mobility Management Framework}\label{secgsi}
In this section, we develop a framework that supports online MM requiring only causal information. Specifically, when making the MM decisions for task $m$, the user has no information about tasks $m+1$, $m+2$, ... . We will prove that our proposed algorithm achieves close-to-optimal performance compared with the oracle algorithm with $J$-step lookahead. The information regarding task $m$ can be classified into two categories depending on which entity possesses the information:
\begin{itemize}
	\item \textbf{User-Side State Information}: The user's location $L_m$, the available candidate BSs $\mathcal{A}(L_m)$,  the input data size $\lambda_m$ and the computation intensity $\gamma_m$.
	\item \textbf{BS-Side State Information}: For each BS $n \in \mathcal{A}(L_m)$, the allocated CPU frequency $f_{m,n}$, the uplink channel gain $H_{m,n}$ and the inter-cell interference $I_{m,n}$.
\end{itemize}

%

Depending on whether the user has the BS-side state information, we will consider two deployment scenarios. In the first scenario, the user knows both the user-side state information and BS-side state information exactly, i.e., the user has Global State Information (GSI). In the second scenario, the user only has the user-side state information, i.e., the user has Local State Information (LSI). In this case, the user needs to learn the BS-side state information in order to make proper MM decisions.


\subsection{EMM-GSI Algorithm}
In this subsection, we present online MM framework for the scenario with GSI.
Assume that the serving BS set does not change during one task, then it is clear that if the user has GSI, radio handover and computation migration of subtasks can be avoided. It is straightforward for the user to select the best BS for offloading and computation and stick to the BS for the entire task. Therefore, for each task $m$, all the subtasks are served by the optimal BS $a_m^*$, i.e., $a_m^1=a_m^2=...=a_m^{K_m}=a_m^*$. 
We use $D(m, n)$ to denote the overall delay and $E(m, n)$ to denote the overall energy consumption by associating to BS $n$ for task $m$ with GSI.

However, a significant challenge remains in directly solving \textbf{P1} since the long-term energy consumption budget couples the MM decisions across different tasks: using more energy for the current task will potentially reduce the energy budget available for future uses, and yet the decisions have to be made without foreseeing the future. To address this challenge, we leverage Lyapunov optimization technique which enables us to solve a deterministic problem for each task with low complexity, while adaptively balancing the delay performance and energy consumption over time.

To guide the MM decisions with Lyapunov optimization technique, we first construct a virtual energy deficit queue. Specifically, the energy deficit queue evolves as
\begin{align}\label{queue}
q(m+1) = \max \{q(m) + E(m, a_m^*) - \alpha B/M, 0\},
\end{align}
with $q(0)=0$. The virtual queue length $q(m)$ indicates how far the current energy usage deviates from the battery energy budget. Since the battery capacity of the user device is finite, it is necessary to consider the case with finite tasks and propose an approach that can guarantee the worst-case delay performance over the finite time horizon. Moreover, both the user-side state information and BS-side state information may not follow a well-defined stochastic process. Therefore, we do not make any ergodic assumptions on the state information. Instead, we adopt a non-ergodic version of Lyapunov optimization, which applies to any arbitrary sample path of the task and system dynamics. The algorithm is called EMM-GSI, as shown in Algorithm 1.

%

%
%

\begin{algorithm}
	\caption{EMM-GSI Algorithm}
	\begin{algorithmic}[1]
		\State \textbf{Input}: $L_m$, $\mathcal{A}(L_m)$, $\lambda_m$, $\gamma_m$, and $\forall n \in \mathcal{A}(L_m)$, $f_{m,n}$, $H_{m,n}$, $I_{m,n}$ at the beginning of offloading each task $m$.
		\If{$m = rJ+1, \forall r = 0,1,...,R-1$}
		\State $q(m) \leftarrow 0$ and $V \leftarrow V_r$.
		\EndIf
		\State Choose $a_m^*$ subject to \eqref{maxdelay}, \eqref{coverage} by solving
		\begin{align}
		(\textbf{P3})~~\min_{n\in\mathcal{A}(L_m)}~VD(m, n) + q(m)E(m, n). \nonumber
		\end{align}
		\State Update $q(m)$ according to \eqref{queue}.
	\end{algorithmic}
\end{algorithm}

Note that EMM-GSI algorithm works in an online fashion, because it only requires the currently available information as the inputs. $V_0, V_1, ..., V_{R-1}$ is a sequence of positive control parameters to dynamically adjust the tradeoff between delay performance and energy consumption over the $R$ frames, each with $J$ periods. Lines 2 - 4 reset the energy deficit virtual queue at the beginning of each frame. Line 5 defines an online optimization problem \textbf{P3} to decide the MM decisions for each task, which is a minimum seeking problem with computational complexity $O(|\mathcal{A}(L_m)|)$, where $|\mathcal{A}(L_m)|$ is the number of candidate BSs for task $m$.
The optimization problem aims to minimize a weighted sum of the delay cost and energy consumption where the weight depends on the current energy deficit queue length and is varying over time. A large weight will be placed on the energy consumption if the current energy deficit is large. The energy deficit queue maintains without foreseeing the future, thereby enabling online decisions. Note that since there is no radio handover and computation migration, \textbf{P3} is equivalent to
\begin{align}
\min_{n\in\mathcal{A}(L_m)} Vd(m, n) + q(m)e(m, n).
\end{align}
Conveniently, we write $z(m,n) \triangleq Vd(m,n) + q(m)e(m,n)$.

\subsection{Performance Bound}
In this subsection, we present the performance analysis of the EMM-GSI algorithm. Under the feasibility assumption that there exists at least one solution to \textbf{P2}, Theorem 1 provides the performance guarantee of EMM-GSI algorithm.

\begin{theorem}
	For any fixed integer $J \in \mathbb{Z}_+$ and $R \in \mathbb{Z}_+$ such that $M = RJ$, the following statements hold.	
	
	(1) The average delay performance achieved by EMM-GSI algorithm satisfies:
	\begin{align}
	d^*_G \leq \frac{1}{R}\sum_{r=0}^{R-1}g^*_r + \frac{UJ}{R}\sum_{r=0}^{R-1}\frac{1}{V_r},
	\end{align}
	where $g^*_r$ is the optimal average delay of the $J$-step lookahead problem for frame $r$, and $U$ is a constant defined as $	U \triangleq \frac{1}{2} \max\{(E(m, a_m^*) - \alpha B/M)^2\} \label{constant}$.

	(2) The total energy consumption is within a bounded deviation:
	\begin{align}
	e^*_G\leq \alpha B + \sum_{r=0}^{R-1} \sqrt{2UJ^2 + 2V_rJg^*_r}	.
	\end{align}
\end{theorem}
\begin{proof}
		See Appendix~\ref{fsi}.
\end{proof}

Theorem 1 shows that using the proposed EMM-GSI algorithm, the worst-case average delay is no more than $O(1/V)$ with respect to the optimal average delay achieved by the $J$-step lookahead problem. Meanwhile, the energy consumption is within a bounded deviation $O(V)$  compared to the given energy budget. Hence, there exists a delay-energy tradeoff of $[O(1/V),O(V)]$. By adjusting $V$, we can balance the average delay and energy consumption.

\section{Learning with LSI Only} \label{seclsi}
In this section, we consider the scenario that the user has LSI only.
We augment our EMM algorithm with online learning based on the MAB framework in order to learn the optimal BS (i.e. the solution to \textbf{P3}) without initially requiring the BS-side information. Learning the optimal BS incurs additional costs since (1) suboptimal BSs will be selected during the learning process, and (2) radio handover and computation migration is inevitable. We also provide theoretical bounds on the performance loss of the proposed algorithm due to learning.

\subsection{EMM-LSI Algorithm}
When the user has only LSI, MM is much more difficult since there is no \textit{a priori} information about which BS provides the best delay performance while incurring less energy consumption. Specifically, the user cannot directly solve \textbf{P3} since $d(m, n)$ and $e(m, n)$ rely on BS-side information such as $f_{m,n}$, $H_{m,n}$ and $I_{m,n}$, which are unknown. Thus the user has to learn the optimal BS on-the-fly.

A straightforward learning scheme is as follows: the user offloads one subtask of task $m$ to every BS $n$ in $\mathcal{A}(L_m)$ and observes the computation delay $\tilde{d}(m, n)$ and energy consumption $\tilde{e}(m, n)$ (and hence the observed $\tilde{z}(m, n)=V\tilde{d}(m, n)+q(m)\tilde{e}(m, n)$). If observations are accurate, namely $\tilde{d}(m, n) = d(m,n)$ and $\tilde{e}(m, n)=e(m,n)$ (and hence $\tilde{z}(m, n) = z(m, n)$), then learning can be terminated and the remaining $K_m - |\mathcal{A}(L_m)|$ subtasks of task $m$ will be offloaded to the BS that is the solution to $\min_{n} \tilde{z}(m, n)$. However, due to the variance in computation intensity, wireless channel state and many other factors, $\tilde{z}(m, n)$ is only a noisy version of $z(m, n)$. In the presence of such measurement variance, this simple learning algorithm can perform very poorly since the user may get trapped in a BS whose $z(m, n)$ is actually large. Therefore, a more sophisticated and effective learning algorithm requires continuous learning to smooth out the measurement noise. In fact, MM with only LSI manifests a classic sequential decision making problem that involves a critical tradeoff between exploration and exploitation: the user needs to explore the different BSs by offloading subtasks to them in order to learn good estimates of $z(m, n), \forall n \in \mathcal{A}(L_m)$, while at the same time it wants to offload as many subtasks as possible to the \textit{a priori unknown} optimal BS.

Sequential decision making problems under uncertainties have been studied under the MAB framework and efficient learning algorithms have been developed that provide strong performance guarantee. In this paper, we augment our EMM algorithm with the so-called UCB1 algorithm \cite{auer2002finite} to learn the optimal BS. Specifically, UCB1 is an index-based algorithm, which assigns an index to each candidate BS and updates the indices of the BSs as more subtasks of a task have been offloaded. Then the next subtask will be offloaded to the BS with the largest index. The index for a BS $n \in \mathcal{A}(L_m)$ is in fact an upper confidence bound on the empirical estimate of $z(m, n)$. Nevertheless, learning algorithms other than UCB1 can also be incorporated in our framework.

The EMM-LSI algorithm is shown in Algorithm 2. The major difference from Algorithm 1 is that instead of solving \textbf{P3} exactly, we use the UCB1 algorithm as a subroutine to learn the optimal BS to minimize the objective in \textbf{P3}, which is reflected from Lines 5 through 15. Let $\bar{z}_{m,n,k}$ denote empirical sample-mean estimate of $z(m, n)$ after the first $k$ subtasks have been offloaded and their corresponding delay and energy performance have been measured. We use $\theta_{m,n,k}$ to denote the number of subtasks that have been offloaded to BS $n$ up to subtask $k$. Lines 5-9 is the initialization phase, and Lines 10-15 is the continuous learning phase. The decision making problem for each subtask is a minimum seeking problem with computational complexity $O(|\mathcal{A}(L_m)|)$, thus for each task, the computational complexity of the EMM-LSI algorithm is $O(K_m |\mathcal{A}(L_m)|)$.

\begin{algorithm}
	\caption{EMM-LSI Algorithm}
	\begin{algorithmic}[1]
		\State \textbf{Input}: $L_m$, $\mathcal{A}(L_m)$, $\lambda_m$, $\gamma_m$ at the beginning of offloading each task $m$.
		\If{$m = rJ+1, \forall r = 0,1,...,R-1$}
		\State $q(m) \leftarrow 0$ and $V \leftarrow V_r$.
		\EndIf
		\For {$k = 1,...,|\mathcal{A}(L_m)|$}  \Comment{\textit{UCB1 Learning}}
		\State Connect to each BS $n\in\mathcal{A}(L_m)$ once.
		\State Update $\bar z_{m,n,k}=V\tilde d(m,n) + q(m)\tilde e(m,n)$.
		\State Update $\theta_{m,n,k}=1$.
		\EndFor
		\For{$k =|\mathcal{A}(L_m)|+1,...,K_m$}
		\State Connect to $a_m^k= \arg\min_{n} \left\{\bar z_{m,n,k}- \beta\sqrt{\frac{2\ln k}{\theta_{m,n,k}}}\right\}$.
		\State Observe $\tilde d(m,a_m^k)$ and $\tilde e(m,a_m^k)$.
		\State $\bar z_{m,a_m^k,k} \leftarrow \frac{\theta_{m,a_m^k,k}\bar{z}_{m,a_m^k,k} + V\tilde d(m,a_m^k) + q(m)\tilde e(m,a_m^k)}{\theta_{m,a_m^k,k}+1}$.
		\State $\theta_{m,a_m^k,k} \leftarrow \theta_{m,a_m^k,k} + 1$.
		\EndFor
		\State Update $q(m)$ according to \eqref{queue}.
	\end{algorithmic}
\end{algorithm}

\subsection{Algorithm Performance}
In this subsection, we analyze the performance of EMM-LSI. We first bound the gap between the exact solution of \textbf{P3} with GSI and the UCB1 learning algorithm with LSI for each task. 
We adopt the concept of \emph{learning regret} to measure the performance loss for each task due to learning, which is commonly used in the MAB framework \cite{auer2002finite}. Formally, the learning regret is defined as follows
\begin{align}
R_m = \mathbb{E}[Z(m, \a_m) - Z(m, a_m^*)],
\end{align}
where $Z(m, \a_m) = VD(m, \a_m) + q(m)E(m, \a_m)$ is the weighted cost achieved by the sequence of MM decisions $\a_m$ resulted from UCB1, and $Z(m, a_m^*) = VD(m, a_m^*) + q(m)E(m, a_m^*)$ is achieved by always connecting to the optimal BS $a_m^*$ that solves \textbf{P3}.

Although the learning regret of the UCB1 algorithm has been well understood, characterizing that in our setting faces new challenges: the learning regret is a result of not only offloading subtasks to suboptimal BSs, but also radio handover and computation migration. Specifically, the learning regret can be decomposed into two terms \cite{agrawal1988asymptotically}, namely the \emph{sampling regret} and the \emph{handover regret}:
\begin{align}
R_m &= \underbrace{\mathbb{E}\left[\sum_{k=1}^{K_m} z(m, a_m^k) - Z(m, a_m^*)\right]}_{\text{sampling regret}} + V\underbrace{\mathbb{E}\left[h(m, \a_m)\right]}_{\text{handover regret}}.
\end{align}

We provide an upper bound on the learning regret of UCB1 considering the handover regret in the following proposition.
\begin{proposition}
	For task $m$ comprising $K_m$ subtasks, the learning regret $R_m$ is upper bounded as follows:
	\begin{align}
	R_m(K_m)\leq &\beta \left[8\sum_{n\neq a_m^*}\frac{\ln K_m}{\delta_{m,n}} + \left(1+ \frac{\pi^2}{3}\right)\sum_{n\neq a_m^*}\delta_{m,n} \right] \nonumber\\
	& +VC_m \left[2\sum_{n\neq a_m^*} \left(\frac{8\ln K_m}{\delta_{m,n}^2} + 1+ \frac{\pi^2}{3} \right)+1 \right] , \label{prop1}
	\end{align}
	where $\beta = \sup_n \tilde{z}(m, n)$ and $\delta_{m,n}=(Z(m,n)-Z(a_m^*)) / \beta  K_m$.
\end{proposition}
\begin{proof}
		See Appendix~\ref{ucb1}.
\end{proof}

\begin{remark}
	Parameter $\beta$ is used to normalize the utility function. In real implementations, it is difficult to obtain the exact value of $\beta$ due to lack of the BS-side state information. However, a reasonably good estimate of $\beta$ can be obtained based on the history data, e.g., setting $\beta$ as the maximum $\tilde{z}(m, n)$ that has been observed.
\end{remark}

The bound on the learning regret established in Proposition 1 is logarithmic in the number of subtasks $K_m$. It also implies that \textbf{P3} can be approximately solved by UCB1 within a bounded deviation, denoted by $W$, since $K_m$ is upper bounded by $\bar{K}$. The performance of EMM-LSI can then be expressed in Theorem 2.

\begin{theorem}
	For any fixed integer $J \in \mathbb{Z}_+$ and $R \in \mathbb{Z}_+$ such that $M = RJ$, the following statements hold.	
	
	(1) The average delay performance achieved by EMM-LSI algorithm satisfies:
	\begin{align}
	d^*_L \leq \frac{1}{R}\sum_{r=0}^{R-1}g^*_r + \frac{UJ+W}{R}\sum_{r=0}^{R-1}\frac{1}{V_r}.
	\end{align}
	
	(2) The total energy consumption is within a bounded deviation:
	\begin{align}
	e^*_L \leq \alpha B + \sum_{r=0}^{R-1} \sqrt{2[UJ^2 + V_rJg^*_r+WJ]}.	
	\end{align}
\end{theorem}
\begin{proof}
		See Appendix~\ref{lsi}.
\end{proof}

Theorem 2 shows that the proposed EMM-LSI algorithm can provide a strong performance guarantee: even if the user cannot acquire the exact BS-side state information, the average delay performance can still be guaranteed through the proposed algorithm, while the energy consumption is within a bounded deviation from the given energy budget.

\subsection{Implementation Considerations} \label{Kstop}
In the proposed EMM-LSI algorithm, the user keeps learning the optimal BS while offloading all $K_m$ subtasks of task $m$. Although Proposition 1 provides an upper bound on the performance loss due to continuous learning, in practice, the loss can be large when the one-time handover cost is relatively large. For instance, when the second-best BS has a similar value of $z(m, n)$ as the optimal BS, the UCB1 algorithm can keep alternating between these two BSs for many subtasks, thereby incurring a significant handover and migration cost. To circumvent this issue, there are two possible heuristic schemes. 

1) The first scheme stops learning after a pre-determined finite number $K_s$ of times of subtask offloading. That is, UCB1 is applied only for the first $K_s$ subtasks. The remaining $K_m - K_s$ subtasks, if any, will all be offloaded to the BS with the lowest value of $\bar{z}_{m, n, K_s}$. Clearly, there is a tradeoff for deciding $K_s$: if $K_s$ is too small, the probability that a suboptimal BS is regarded as the optimal is high, and hence, leading to a large cost for offloading the remaining subtasks to the suboptimal BS. On the other hand, if $K_s$ is too large, a large handover cost may be incurred. We will quantify this tradeoff in our simulation results.

2) The second scheme stops learning when the best and second-best BSs have very similar performance. Specifically, the stopping criteria is
\begin{align}
&\bar{z}_{m, n^*, k} - \bar{z}_{m, n^{\dagger}, k} \leq \epsilon\\
&\theta_{m, n^*, k} \geq K_0, \theta_{m, n^\dagger, k} \geq K_0 ,
\end{align}
where $n^*$ represents the learned best BS and $n^\dagger$ is the learned second-best BS so far, and $\epsilon, K_0$ are pre-determined parameters.

\section{Varying BS Set} \label{seclsid}
In this section, we consider a more general setting in which the set of candidate BSs during the offloading of one task can vary. For example, BSs are turned on/off according to the BS sleeping strategy for energy saving purposes \cite{7114343} or small cell owner-governed processes. We develop a modified version of the EMM-LSI algorithm, called EMM-LSI-V, based on the VMAB framework and characterize its performance.

\subsection{EMM-LSI-V Algorithm}
The varying set of BSs creates a big challenge in learning the optimal BS that solves \textbf{P3}. With the conventional UCB1 algorithm, the user has to restart the learning process whenever a new BS appears. Apparently, this learning strategy is very inefficient since it simply restarts the learning process without reusing what has been learned. Although the available BS set changes, the states of other BSs are likely to remain the same. Therefore, proper learning algorithms that effectively reuse the already learned information are needed. 

To efficiently learn the optimal BS among a varying BS set, we adopt the VMAB framework \cite{bnaya2013social}, in which BSs can appear or disappear unexpectedly with unknown lifespan. Define an epoch as the interval in which the available BS set is invariant, and let $B_m$ be the total number of epochs for task $m$, which is unknown in advance. Note that $B_m=1, \forall m$ corresponds to the case that we considered in Section \ref{seclsi}.
The available BS set for epoch $b=1,2,...,B_m$ is denoted as $\mathcal{A}_{m,b}$ and let $\mathcal{A}_m$ be the union of $\mathcal{A}_{m,b}, \forall b=1,...,B_m$. To simplify the problem, we assume that each BS only appears once during each task. If a BS appears for the second time, it can be treated as a new BS. For each BS $n \in \mathcal{A}_{m}$, the lifespan is denoted as $[u_n, v_n]$ with $1\leq u_n, v_n \leq K_m$, which indicates that BS $n$ is present from subtask $u_n$ through subtask $v_n$. We also denote $K_{m, b}$ as the total number of subtasks of task $m$ completed by the end of epoch $b$. Clearly, $K_{m,B_m}=K_m$.  

The EMM-LSI-V algorithm developed on volatile UCB1 (VUCB1) learning is proposed in Algorithm 3. In VUCB1 learning, a UCB1-like algorithm is implemented for each epoch. The differences are two-fold. First, the initialization for each epoch (Lines 6-10) only applies to the newly appeared BSs, while the information for the remaining BSs is retained and hence reused. Second, the index term on Line 12 used to guide the subtask offloading decision takes into account the appearance time of the BS.

\begin{algorithm}
	\caption{EMM- LSI-V Algorithm}
	\begin{algorithmic}[1]
		\State \textbf{Input}: $L_m$, $\lambda_m$, $\gamma_m$ at the beginning of offloading each task $m$.
		\If{$t = rJ+1, \forall r = 0,1,...,R-1$}
		\State $q(m) \leftarrow 0$ and $V \leftarrow V_r$.
		\EndIf
		\For {$k=1,...,K_m$}  \Comment{\textit{VUCB1 Learning}}
		\If { $k$ is the first block of an epoch}
		\State \textbf{Input}: $\mathcal{A}_{m,b}$
		\State Connect to each first appeared BS $n\in\mathcal{A}_{m,b}$ once.
		\State Update $\bar z_{m,n,k}=V\tilde d(m,n) + q(m)\tilde e(m,n)$.
		\State Update $\theta_{m,n,k}=1$.		
		\Else
		\State $a_m^k= \arg\min_{n} \left\{\bar z_{m,n,k}- \beta\sqrt{\frac{2\ln (k-u_n)}{\theta_{m,n,k}}}\right\}$, connect to BS $a_m^k$.
		\State Observe $\tilde d(m,a_m^k)$ and $\tilde e(m,a_m^k)$.
		\State $\bar z_{m,a_m^k,k} \leftarrow \frac{\theta_{m,a_m^k,k}\bar{z}_{m,a_m^k,k} + V\tilde d(m,a_m^k) + q(m)\tilde e(m,a_m^k)}{\theta_{m,a_m^k,k}+1}$.
		\State $\theta_{m,a_m^k,k} \leftarrow \theta_{m,a_m^k,k} + 1$.
		\EndIf
		\EndFor		
		\State Update $q(m)$ according to \eqref{queue}.
	\end{algorithmic}
\end{algorithm}

\subsection{Algorithm Performance}
We characterize the performance of the VUCB1 learning as follows. Let $a_{m,b}^*$ as the optimal BS at epoch $b$ for task $m$. The learning regret is thus
\begin{align}
R_m = &\underbrace{\sum_{b=1}^{B_m} \mathbb{E}\left[\sum_{k=K_{m, b-1}+1}^{K_{m, b}}z(m, a_m^k)- Z(m, a_{m,b}^*)\right]}_{\text{sampling regret}} \nonumber\\
&+ V\underbrace{\mathbb{E}\left[h(m, \a_m)\right]}_{\text{handover regret}}.
\end{align}

\begin{proposition}
	For task $m$ comprising $K_m$ subtasks, if there are $B_m$ epochs, the total regret $R_m$ of VUCB1 is of $O(B_m\ln K_m)$.
\end{proposition}
\begin{proof}
		See Appendix~\ref{vucb1}.
\end{proof}

Proposition 2 states that VUCB1 learning can provide a bounded deviation, defined as $W'$, from exactly solving \textbf{P3}. Therefore, our EMM-LSI-V algorithm can still provide strong performance guarantee by substituting the bounded deviation $W$ with $W'$  in Theorem 2. 

\section{Simulations} \label{secsim}
In this section, we evaluate the average delay performance and total energy consumption of the proposed EMM algorithms and verify the theoretical results through simulations using MATLAB. We simulate a 1km$\times$1km square area with 49 BSs deployed on a regular grid network. The user can associate with BSs within a radius of 150m. The user trajectory is generated by the random walk model. The wireless channel gain is modeled as $H_{m,n}= 127 + 30\times\log d$, as suggested in \cite{7842625}. Besides, channel bandwidth $W = 20 \text{MHz}$, noise power $\sigma^2 =2\times10^{-13}\text{W}$, and transmit power $P_{\mathrm{tx}}=0.5\text{W}$.


 We consider an application of video stream analysis with totally $M=500$ video tasks generated during the entire trip. Each subtask is a one-second video clip. According to \cite{anjum2016video}, we set $\lambda_0=0.62 \text{Mbits}$, which is the data size of a one-second QCIF format video with $176\times144$ video resolution, $24.8\text{k}$ pixels per frame and 25 fps (frame per second). Each video is set to be $1\text{min}$ to $2\text{min}$ long, i.e., $K_m$ is uniformly selected from $\{60,61,...,120\}$, thus $\lambda_m\in [37.2, 74.4]~\text{Mbits}$. Each subtask has completion deadline $150 \text{ms}$, and the computation intensity $\gamma_m$ is uniformly distributed within $[500, 1000]~\text{cycles/bit}$. 
 Each MEC sever is equipped with multiple CPU cores, and the sum frequency $F_n=25\text{GHz}$. The available computation capability for each task follows uniform distribution with $f_{m,n} \in [0, F_n]~ \text{GHz}$. In addition, one-time handover cost $C_m=5\text{ms}$, and battery capacity $B=1000\text{J}$.

We introduce four benchmark algorithms to evaluate the performance of the proposed EMM algorithms: 
1) \textbf{$J$-step Lookahead}: this is the oracle benchmark described in Section \ref{jstep}. We set $J = 5$ and thus $R =M/J= 100$. Note that solving the $J$-step lookahead problem is extremely computationally complex.
2) \textbf{Delay Optimal (GSI)}: the user always associates with the BS with the lowest delay and disregards the energy consumption constraint. 
3) \textbf{Energy Optimal (GSI)}: the user always associates with the BS with the best channel condition without considering the delay performance. In fact, this is the standard 3GPP LTE handover protocol with Event A3 handover condition where the handover offset is set to be zero (see \cite{ts36331}, Sec. 5.5.4). Both delay optimal and energy optimal benchmarks are implemented in the GSI scenario.
4) \textbf{Radio-LSI}: this benchmark learns the BS with best channel condition based on the MAB theory \cite{shen2017learning}. It is implemented in the LSI scenario to compare with the EMM-LSI algorithm.

\begin{figure}[!htb]
	\centering
	\vspace{-0.2in}
	\subfigure[Average delay]{\label{Performance_E2M21}			
		\includegraphics[width=0.4\textwidth]{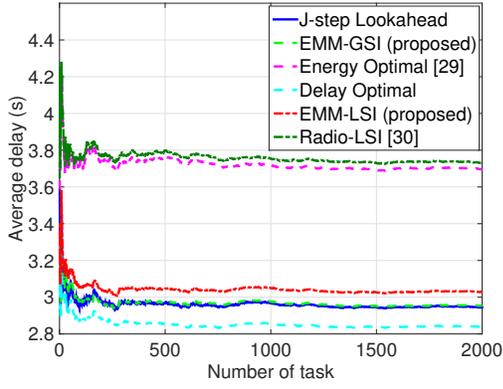}}
	\subfigure[Total energy consumption]{\label{Performance_E2M22}	
		\includegraphics[width=0.4\textwidth]{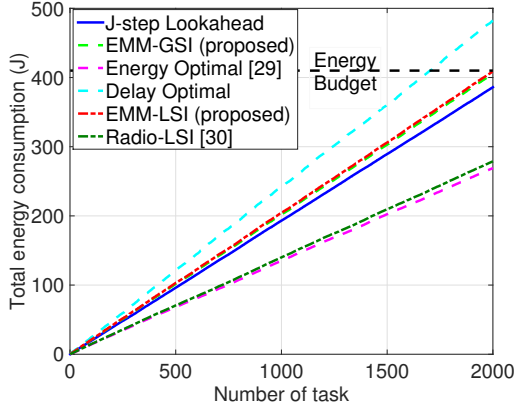}}
	\caption{Performance of EMM ($V=0.01$, $\alpha B=410\text{J}$, $K_s=20$, $30\%$ observation variance).}
	\label{Performance_E2M2}
	\vspace{-0.1in}
\end{figure}

Fig. \ref{Performance_E2M2} compares the average delay performance and total energy consumption over the $M$ tasks of EMM-GSI, EMM-LSI and four benchmark algorithms. Here we set $30\%$ observation variance in the LSI scenario and let EMM-LSI algorithm stop learning after offloading $K_s=20$ subtasks to avoid frequent radio handover and computation migration, as discussed in Section \ref{Kstop}. As can be seen, our two EMM algorithms satisfy the energy consumption constraint while keeping the delay low. In the GSI scenario, EMM-GSI algorithm effectively balances delay and energy consumption and achieves the delay close to the J-step Lookahead. EMM-LSI algorithm is just slightly worse than EMM-GSI algorithm. Compared with the Radio-LSI algorithm, EMM-LSI algorithm performs better in delay performance since it learns both radio and computation states rather than only the wireless channel condition.


\begin{figure}[!htb]
	\centering
	\vspace{-0.1in}
	\includegraphics[width=0.4\textwidth]{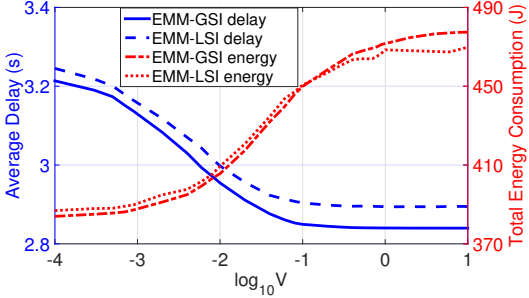}
	\caption{Impact of $V$ ($\alpha B=410\text{J}$, $K_s=20$, $30\%$ observation variance).}	\label{Performance_V}
\end{figure}

Fig. \ref{Performance_V} shows the impact of control parameter $V$ on the average delay and total energy consumption. By increasing $V$ from $10^{-4}$ to $10$, both EMM-GSI and EMM-LSI algorithms care more about the delay performance, and thus the average delay decreases. However, with less concern on the energy consumption, the total energy consumption increases and will finally exceed the given budget. The delay-energy performance follows the  $[O(1/V),O(V)]$ tradeoff, which verifies Theorem 1 and Theorem 2. Meanwhile, the results also provide guidelines for selecting $V$ in real implementations: under the energy budget constraint, one should choose appropriate $V$ that can minimize the average delay performance.

\begin{figure}[!htb]
	\centering
	\vspace{-0.2in}
	\subfigure[Average delay]{\label{Performance_B1}
		\includegraphics[width=0.4\textwidth]{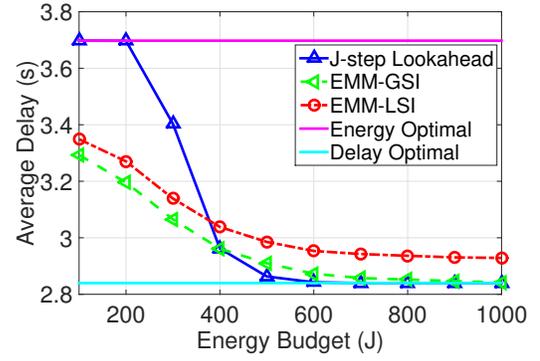}}
	\subfigure[Total energy consumption]{\label{Performance_B2}
		\includegraphics[width=0.4\textwidth]{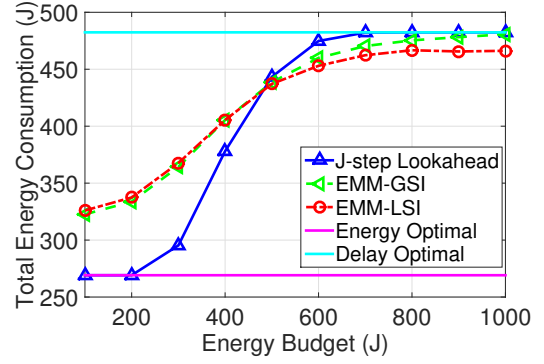}}
	\caption{Impact of energy budget $\alpha B$  ($V=0.01$, $K_s=20$, $30\%$ observation variance).}
	\label{Performance_B}
\end{figure}

By varying the energy capping parameter $\alpha$ from $10\%$ to $100\%$, we explore the impact of energy budget on the average delay and total energy consumption, as shown in Fig. \ref{Performance_B}. When the energy budget is large, EMM-GSI achieves the optimal delay since the energy constraint is always satisfied, while EMM-LSI incurs additional performance loss due to the learning process. When the energy budget is too low, there is possibly no feasible solution, thus the energy constraint is violated. In between, both EMM algorithms can tradeoff between the average delay and energy consumption, and the performance of EMM-GSI is very close to the $J$-step Lookahead.

\begin{figure}[!htb]
	\centering
	\vspace{-0.2in}
	\subfigure[Probability of connecting suboptimal BS after learning]{\label{Performance_K1}
		\includegraphics[width=0.4\textwidth]{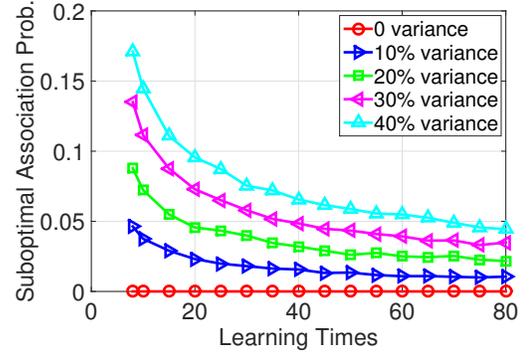}}
	\subfigure[Average delay]{\label{Performance_K2}
		\includegraphics[width=0.4\textwidth]{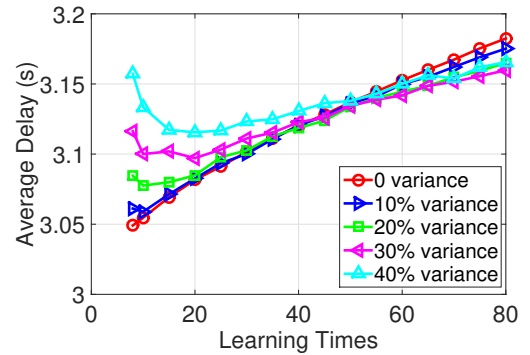}}
	\caption{Impact of learning times $K_s$ ($V=0.01$, $\alpha B=410\text{J}$).}
	\label{Performance_K}
\end{figure}

For implementation considerations, the impact of the number of subtasks $K_s$ used for learning in EMM-LSI algorithm is further evaluated. 
We set $K_s$ to vary from $8$ to $80$, carry out simulations under different observation variance, and repeat 10 times for average. 
Fig. \ref{Performance_K1} shows the probability of connecting to a suboptimal BS after using $K_s$ subtasks to learn. When there is no observation variances, the user can always select the optimal BS after connecting to each available BS once. When the observation variance increases, the probability of connecting to a suboptimal BS increases. However, as $K_s$ increase, the probability of connecting to a suboptimal BS decreases drastically. 
Fig. \ref{Performance_K2} shows the impact of $K_s$ on the average delay. With $K_s$ increasing, the average delay decreases first and then increases, except for the case with zero variance where learning always increases the regret. 
This is because when $K_s$ is small, the probability of connecting to a suboptimal BS after learning is large, which leads to high additional cost. When $K_s$ is large, the frequent handover increases the handover regret and thus degrades the delay performance.
Therefore, learning time $K_s$ should be carefully selected to balance the aforementioned two factors. For example, in our settings, under $30\%$ observation variance, $K_s=20$ can obtain the best delay performance.

\begin{table}[!htb]
	\caption{Available BSs and normalized utility}
	\label{normalized_utility}
	\centering
	\begin{tabular}{||c||c|c|c|c|c||}
		\hline
		Index of BS & 1 & 2 & 3&4&5\\
		\hline
		Normalized utility & 0.5&0.8&0.4&0.9&0.7\\
		\hline
		Epoch 1 &${\surd}$ & ${\surd}$ & -- & --&-- \\
		\hline
		Epoch 2 &${\surd}$& ${\surd}$ &  ${\surd}$  &${\surd}$ &-- \\
		\hline
		Epoch 3 & ${\surd}$&${\surd}$& $\times$ & ${\surd}$  &${\surd}$ \\
		\hline
	\end{tabular}
\end{table}

\begin{figure}[!htb]
	\centering
	\subfigure[Average utility]{\label{Performance_VUCB1}
		\includegraphics[width=0.4\textwidth]{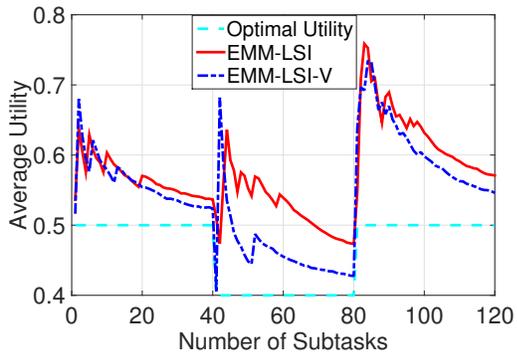}}
	\hspace{-0.1in}
	\subfigure[Handover times]{\label{Performance_VUCB2}
		\includegraphics[width=0.4\textwidth]{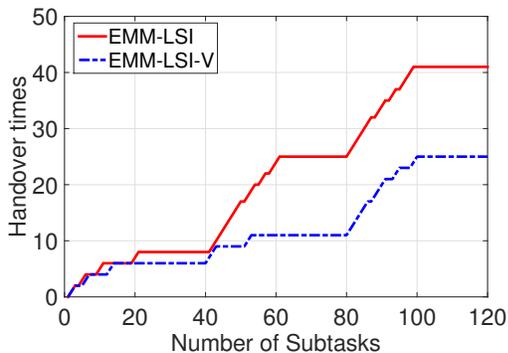}}
	\hspace{-0.1in}
	\caption{EMM-LSI-V algorithm vs. EMM-LSI algorithm.}
	\label{Performance_VUCB}
\end{figure}

Finally, we compare the proposed EMM-LSI-V algorithm with EMM-LSI under the dynamic BS set. We illustrate the results by dividing one task into 3 epochs. The available BSs and their normalized utility (defined in \textbf{P3}, which reflects both the delay performance and energy consumption) are shown in Table \ref{normalized_utility}. In epoch $2$, there appears an optimal BS and a suboptimal BS, while in epoch $3$, an optimal BS disappears and a suboptimal BS appears. Each epoch has 40 subtasks and $K_s=20$. 
As shown in Fig. \ref{Performance_VUCB}, $K_s=40$ and $K_s=80$ are the beginning of epoch 2 and epoch 3, thus both algorithms start to learn the environmental change and the average utility suffers sudden increases. However, the EMM-LSI-V algorithm converges faster than EMM-LSI algorithm does, while efficiently reduces the handover times. This is because EMM-LSI-V algorithm is able to retain the information of remaining BSs while EMM-LSI algorithm restarts the learning process whenever there is a change of the BS set.

\section{Conclusions} \label{seccon}
In this paper, we studied the MM problem for MEC-enabled UDN. We developed a novel user-centric MM framework and designed MM algorithms, called EMM, that can be applied to both GSI and LSI scenarios by integrating Lyapunov optimization and MAB techniques. Taking radio handover and computation migration cost into consideration, we proved that our proposed algorithms can optimize the delay performance while approximately satisfying the energy consumption budget of the user. Furthermore, we proposed a generalized EMM algorithm that can handle varying BS sets based on the VMAB framework. 
Simulations show that our proposed EMM algorithm can achieve close-to-optimal delay performance while satisfying the energy consumption constraint of the user.
Future research directions include designing MM schemes for high mobility scenarios where the user may move a lot during the processing of a task, and considering cooperative computing among BSs. 


\appendices{}


\section{Proof of Theorem 1} \label{fsi}
For notational convenience, we define $y(m) = E(m, a_m^*) - \alpha B/M$. According to the energy deficit queue in \eqref{queue}, it is easy to see
\begin{align}
q(m+1) - q(m) \geq y(m).
\end{align}
Summing the above over $m = rJ+1, ..., (r+1)J $, using the law of telescoping sums, we get
\begin{align}
\sum_{t=rJ+1}^{(r+1)J} y(m) \leq q((r+1)J+1)-q(rJ+1) \label{ym},
\end{align}
where $q(rJ+1)=0$ and $q((r+1)J+1)$ is the queue length before reset in frame $r+1$. In what follows, we try to bound $q((r+1)J+1)$.

Define the Lyapunov function $L(q(m))$ as
\begin{align}
	L(q(m)) \triangleq \frac{1}{2}q^2(m).
\end{align}
Moreover, we define the 1-slot Lyapunov drift $\Delta_1(m)$ as:
\begin{align}
\Delta_1(m) = L(q(m+1)) - L(q(m)),
\end{align}
where a ``slot'' refers to the duration of offloading and computation for a task. 

Therefore, the 1-slot drift-plus-penalty function can be expressed as $\Delta_1(m)+VD(m,a_m^*)$, where $V>0$ is a control parameter that affects the tradeoff between delay performance and energy consumption.

 According to the definition of energy deficit queue in \eqref{queue}, squaring the queuing dynamics equation results in the following bound
 \begin{align}
 q^2(m+1) &\leq  (q(m) + y(m))^2 \nonumber\\
 &= q^2(m) + y^2(m) + 2q(m)y(m).
 \end{align}
 
 Therefore, the 1-slot Lyapunov drift $\Delta_1(m)$ satisfies
 \begin{align}
 \Delta_1(m) = L(q(m+1)) - L(q(m)) \leq \frac{1}{2} y^2(m) + q(m) y(m) \label{driftb}.
 \end{align}
 Now define $U$ as a positive constant that upper bounds $\frac{1}{2} y^2(m)$. Such a constant exists under the assumption that $y(m)$ is deterministically bounded. By adding $VD(m,a_m^*)$ at both sides of \eqref{driftb}, we can obtain
 \begin{align}
 &\Delta_1(m)+VD(m,a_m^*)  \nonumber\\
 &	\leq U + VD(m,a_m^*) +q(m)y(m).
 \end{align}

Define the $J$-slot Lyapunov drift as $\Delta_J(rJ) \triangleq L(q((r+1)J+1)) - L(q(rJ+1)) $, we have
\begin{align}
&\Delta_J(rJ) + V \sum_{m=rJ+1}^{(r+1)J}D(m,a_m^*) \\ \label{driftj}
\leq & UJ + V \sum_{m=rJ+1}^{(r+1)J}D(m,a_m^*)+ \sum_{m=rJ+1}^{(r+1)J}q(m)y(m) \nonumber\\
=&UJ + V \sum_{m=rJ+1}^{(r+1)J}D(m,a_m^*) + \sum_{m=rJ+1}^{(r+1)J}q(rJ+1)y(m) \nonumber\\
& + \sum_{m=rJ+1}^{(r+1)J} (q(m) - q(rJ+1)) y(m). \nonumber
\end{align}

Let $y_{max}\geq 0$ denote the maximum positive value of $y(m)$ for all $m$ (otherwise $y_{max}=0$), i.e., $q(m+1)-q(m)\leq y_{max}$. Thus, for $m=rJ+1,...,(r+1)J$,
\begin{align}
q(m)-q(rJ+1)\leq (m-(rJ+1)) y_{max}.
\end{align}
The last term on the right hand side of \eqref{driftj} satisfies
\begin{align}
&\sum_{m=rJ+1}^{(r+1)J} (q(m) - q(rJ+1)) y(m) \nonumber\\
\leq&\sum_{m=rJ+1}^{(r+1)J} (m-(rJ+1)) y^2_{max} \nonumber\\
=&\frac{J(J-1)}{2}y^2_{max} \leq J(J-1)U	.
\end{align}

The right hand side of \eqref{driftj} is bounded by
\begin{align}
&\Delta_J(rJ) + V \sum_{m=rJ+1}^{(r+1)J}D(m,a_m^*) \nonumber\\
\leq& UJ^2 + V \sum_{m=rJ+1}^{(r+1)J}D(m,a_m^*) .  \label{driftjr}
\end{align}

By applying EMM-GSI algorithm on the left-hand side and considering the optimal $J$-step lookahead algorithm on the right-hand side, we obtain the following
\begin{align}\label{V}
\Delta_J(rJ) + V_r \sum_{m=rJ+1}^{(r+1)J}d_G^*(m) \leq UJ^2 + V_r J g^*_r,
\end{align}
where $d_G^*(m)$ is the delay achieved by EMM-GSI algorithm for task $m$.

Therefore,
\begin{align} \label{qr}
q((r+1)J+1) = \sqrt{2\Delta_J(rJ)}\leq \sqrt{2(UJ^2 + V_r J g^*_r )}.
\end{align}

Substituting \eqref{qr} into \eqref{ym}, we have
\begin{align}
\sum_{m=rJ+1}^{(r+1)J} y(m) \leq \sqrt{2(UJ^2 + V_r J g^*_r )}.
\end{align}
Therefore,
\begin{align}
\sum_{m=rJ+1}^{(r+1)J} e_G^*(m) \leq \alpha B/R+\sqrt{2(UJ^2 + V_r J g^*_r )},
\end{align}
where $e_G^*(m)$ is the energy consumption achieved by EMM-GSI algorithm for task $m$.
By summing over $r = 0, 1, ..., R-1$ we prove part (2) of Theorem 1.

By dividing both sides of \eqref{V} by $V_r$, it follows that
\begin{align}
\sum_{m=rJ+1}^{(r+1)J }d_G^*(m) \leq Jg^*_r + \frac{UJ^2}{V_r}.
\end{align}
Thus, by summing over $r = 0, 1, ..., R-1$ and dividing both sides by $RJ$, we prove part (1) of Theorem 1.

\section{Proof of Proposition 1} \label{ucb1}
The proof follows the similar idea of \cite{auer2002finite} and the main difference is that we also bound the handover regret.

Since we only focus on the regret in one task, we omit $m$ for notation convenience. The sampling regret \text{SR} can be written as
\begin{align}
\text{SR}&=\mathbb{E}\left[\sum_{k=1}^{K}z(a^k)- Z(a^*)\right] \nonumber\\
&=\mathbb{E} \left[\sum_{n\in \mathcal{A}(L^m)}\theta_{n,K} \frac{Z(n)}{K} - \theta_{a^*,K} \frac{Z(a^*)}{K}\right] \nonumber\\
&=\sum_{n\neq a^*}\beta\delta_n\mathbb{E}[\theta_{n,K}]	.
\end{align}

We first bound $\theta_{n,K}$. Let $c_{k,s}=\sqrt{2\ln k/s}$, $l$ be any positive integer, and $z'=z/\beta$ is the normalized utility. We have
\begin{align}
	&\theta_{n,K} = 1+\sum_{k=A+1}^{K}\mathbb{I}\left\{a^k=n\right\} \nonumber\\
	&\leq l+ \sum_{k=A+1}^{K}\mathbb{I}\left\{a^k=n, \theta_{n,k-1} \geq l \right\}  \nonumber\\
	&\leq l+\sum_{k=A+1}^{K}\mathbb{I} \left\{ \max _{0<s<k}\bar{z}'_{a^*,s}- c_{k,s} \geq \min_{l\leq s_n<k} \bar{z}'_{n,s_n}-c_{k,s_n}\right\} \nonumber\\
	&\leq l+\sum_{k=1}^{\infty}\sum_{s=1}^{k-1} \sum_{s_n=l}^{k-1} \mathbb{I} \left\{\bar{z}'_{a^*,s}- c_{k,s} \geq \bar{z}'_{n,s_n}-c_{k,s_n}\right\}.
\end{align}

$\mathbb{I} \{\bar{z}'_{a^*,s}- c_{k,s} \geq \bar{z}'_{n,s_n}-c_{k,s_n}\}$ implies that at least one of the following three equations hold
\begin{align}
 \bar{z}'_{a^*,s} &\geq Z(a^*)/\beta K +c_{k,s}, \\
 \bar{z}'_{n,s_n}& \leq Z(n)/\beta K -c_{k,s_n}, \\
   Z(a^*)/\beta K &>Z(n)/\beta K-2c_{k,s_n} .\label{cons1}
\end{align}

By using Chernoff-Hoeffding bound, we have
\begin{align}
&\mathbb{P}\{\bar{z}'_{a^*,s} \geq Z(a^*)/\beta K+c_{k,s}\} \leq e^{-4\ln k}=k^{-4},\\
&\mathbb{P}\{ \bar{z}'_{n,s_n} \leq Z(n)/\beta K -c_{k,s_n} \} \leq k^{-4}.
\end{align}

When $l\geq\lceil \frac{8\ln K}{\delta^2_n}\rceil$, \eqref{cons1} not holds because
\begin{align}
 &~~Z(a^*)\beta K -Z(n)\beta K+2c_{k,s_n}  \nonumber\\
 & =Z(a^*)\beta K -Z(n)\beta K+2\sqrt{2\ln k/s_n} \nonumber \\
 &\leq Z(a^*)\beta K -Z(n)\beta K+\delta_n=0.
\end{align}

Then for any $n \neq a^{*}$, we have
\begin{align}
	\mathbb{E}[\theta_{n,K}] \leq& \left \lceil \frac{8\ln K}{\delta^2_n}\right \rceil+\sum_{k=1}^{\infty}\sum_{s=1}^{k-1} \sum_{s_n=l}^{k-1} \left(\mathbb{P}\{\bar{z}'_{a^*,s} \geq Z(a^*)/\beta K  \right.\nonumber  \\
&\left.+c_{k,s} \}+\mathbb{P}\{ \bar{z}'_{n,s_n} \leq Z(n)/\beta K -c_{k,s_n}\}\right)  \nonumber\\
\leq &\left \lceil \frac{8\ln K}{\delta^2_n}\right \rceil +\sum_{k=1}^{\infty}\sum_{s=1}^{k-1} \sum_{s_n=l}^{k-1} 2k^{-4}\\
\leq &\frac{8\ln K}{\delta^2_n} + 1+ \frac{\pi^2}{3}.
\end{align}

The upper bound of sampling regret is
\begin{align}
\text{SR}&=\sum_{n\neq a^{*}}\beta \delta(n)\mathbb{E}[\theta_{n,K}]	\nonumber\\
& \leq \beta \left[8\sum_{n\neq a^{*}}\frac{\ln K}{\delta_n} + \left(1+ \frac{\pi^2}{3}\right)\sum_{n\neq a^{*}}\delta_n \right].
\end{align}

The upper bound of handover regret is
\begin{align}
\text{HR}& =V\mathbb{E}[h(m, \a_m) ]\nonumber\\
&= VC \mathbb{E}\left[\sum_{k=2}^{K}\mathbb{I}\{a^k\neq a^{k-1}\}\right]\nonumber\\
&=VC\sum_{n\in \mathcal{A}(L^m)}\mathbb{E}\left[\sum_{k=2}^{K}\mathbb{I}\{a^k=n, a^{k-1}\neq n \}\right] .\nonumber\\
\end{align}

Let $S_n=\sum_{k=2}^{K}\mathbb{I}\{a^k=n, a^{k-1}\neq n \}$ count the handover times from BS $n$ to other BSs. Then
\begin{align}
\text{HR}&=VC \left( \sum_{n\neq a^{*}}\mathbb{E}[S_n] + \mathbb{E} [S_{a^{*}}] \right) \nonumber\\
&\leq V C \left(2\sum_{n\neq a^{*}}\mathbb{E}[S_n] +1\right) \leq VC\left(2\sum_{n\neq a^{*}} \mathbb{E} [\theta_{n,K}] + 1 \right) \nonumber\\
&\leq VC \left(2\sum_{n\neq a^{*}}\left[\frac{8\ln K}{\delta^2_n} + 1+ \frac{\pi^2}{3} \right]+1\right ).
\end{align}

By adding \text{SR} and \text{HR}, we prove Proposition 1.

\section{Proof of Theorem 2} \label{lsi}
Let $d_L^*(m)$ and $e_L^*(m)$ be the delay and energy consumption of task $m$ achieved by EMM-LSI algorithm, respectively. From \eqref{prop1}, we get
\begin{align}
	Vd_L^*(m)+q(m)e_L^*\leq Vd_G^*(m)+q(m)e_G^*(m)+W \label{prop1_1}.
\end{align}

 Substituting \eqref{prop1_1} into \eqref{V}, we get
 \begin{align}
 \Delta_J(rJ) + V_r \sum_{m=rJ+1}^{(r+1)J}d_L^*(m) \leq UJ^2 + V_r J g^*_r  +WJ \label{VV}.
 \end{align}

 Thus
\begin{align}
&q((r+1)J+1) = \sqrt{2\Delta_J(rJ)} \nonumber\\
\leq &\sqrt{2[UJ^2 + V_r J g^*_r +WJ]}.
\end{align}

 By \eqref{qr}, we have
 \begin{align}
 \sum_{m=rJ+1}^{(r+1)J} y(m) \leq \sqrt{2[UJ^2 + V_r J g^*_r +WJ]}.
 \end{align}
 By summing over $r = 0, 1, ..., R-1$ we prove part (2) of Theorem 2.

 By dividing both sides of \eqref{VV} by $V_r$, it follows that
 \begin{align}
 \sum_{m=rJ+1}^{(r+1)J }d_L^*(m) \leq Jg^*_r + \frac{UJ^2+WJ}{V_r}.
 \end{align}
 Thus, by summing over $r = 0, 1, ..., R-1$ and dividing both sides by $RJ$, we prove part (1) of Theorem 2.

 \section{Proof of Proposition 2} \label{vucb1}
We only focus on the regret in one task, and thus omit $m$ for notation convenience. We first prove that both the sampling regret and hanover regret in each epoch is $O(\ln K)$.

We first bound the expectation of $\theta_{n,b,K_b} $, which indicates the connection times to an suboptimal BS $n$ in each epoch $b$ after offloading $K_b$ tasks. Let $l$ be any positive integer, $c_{k,s,u}=\sqrt{2\ln (k-u)/s}$. Let $z'=z/\beta$, and $a^*$ be replaced by $a_b^*$, we have
\begin{align}
	&\theta_{n,b,K_b} =\sum_{k=K_{b-1}+1}^{K_b}\mathbb{I}\{a^k=n\} \nonumber\\
	&\leq	l+ \sum_{k=K_{b-1}+1}^{K_b}\mathbb{I}\{a^k=n, \theta_{n,k-1,b} \geq l \}  \nonumber\\
	&\leq	l+ \sum_{k=K_{b-1}+1}^{K_b}\mathbb{I}\left\{ \max _{K_{b-1}<s<k}\bar{z}'_{a^*,s}- c_{k,s,u_{a^*}} \right. \geq \nonumber\\
	&~~~\left. \min_{K_{b-1}+l\leq s_n<k} \bar{z}'_{n,s_n}-c_{k,s_n,u_n}\right\} \nonumber\\
	&\leq l+ \sum_{k=K_{b-1}+1}^{K_b}\sum_{s=K_{b-1}+1}^{k-1} \sum_{s_n=K_{b-1}+l}^{k-1} \mathbb{I}\{\bar{z}'_{a^*,s}- c_{k,s,u_{a^*}}\nonumber\\
	&~~~ \geq  \bar{z}'_{n,s_n}-c_{k,s_n,u_n}\}.
\end{align}

$ \mathbb{I}\{\bar{z}'_{a^*,s}- c_{k,s,u_{a^*}} \geq\bar{z}'_{n,s_n}-c_{k,s_n,u_n}\}$ implies that at least one of the following three equations hold
\begin{align}
 \bar{z}'_{a^*,s} &\geq Z(a^*)/\beta K +c_{k,s,u_{a^*}} ,\\
 \bar{z}'_{n,s_n}& \leq Z(n)/\beta K -c_{k,s_n,u_n} ,\\
 Z(a^*)/\beta K &>Z(n)/\beta K-2c_{k,s_n,u_n} .\label{cons2}
\end{align}

By using Chernoff-Hoeffding bound, we have
\begin{align}
&\mathbb{P}\{\bar{z}'_{a^*,s} \geq Z(a^*)/\beta K+c_{k,s,u_{a^*}} \} \leq (k-u_{a^*})^{-4},\\
&\mathbb{P}\{ \bar{z}'_{n,s_n} \leq Z(n)/\beta K -c_{k,s_n,u_n}\} \leq (k-u_n)^{-4}.
\end{align}
When $l\geq \left \lceil \frac{8\ln (K_b-u_n)}{\delta^2_{n,b}}\right \rceil$, \eqref{cons2} not holds because
\begin{align}
&~~Z(a^*)\beta K -Z(n)\beta K+2c_{k,s_n,u_n}  \nonumber\\
&\leq Z(a^*)\beta K -Z(n)\beta K+\delta_{n,b}=0,
\end{align}
where $\delta_{n,b}=(Z(n)-Z(a_{m,b}^*))/K\beta$.

Then for any $n \neq a^{*}$, we have
\begin{align}
&\mathbb{E}[\theta_{n,b,K_b}] \leq  \nonumber\\
&\left \lceil \frac{8\ln (K_b-u_n)}{\delta^2_{n,b}}\right \rceil+\sum_{k=1}^{\infty}\sum_{s=K_{b-1}+1}^{k-1} \sum_{s_n=K_{b-1}+l}^{k-1}  (\mathbb{P}\{\bar{z}'_{a^*,s} \geq \nonumber\\
&Z(a^*)/\beta K+c_{k,s,u_{a^*}} \} +\mathbb{P}\{ \bar{z}'_{n,s_n} \leq Z(n)/\beta K -c_{k,s_n,u_n}\})  \nonumber\\
&\leq \left \lceil  \frac{8\ln (K_b-u_n)}{\delta^2_{n,b}}\right \rceil +\sum_{k=K_{b-1}+1}^{K_b}\sum_{s=K_{b-1}+1}^{k-1} \sum_{s_n=K_{b-1}+l}^{k-1}\nonumber\\
&~\left((k-u_{a^*})^{-4}+  (k-u_n)^{-4}\right) \nonumber\\
&\leq \frac{8\ln  (K_b-u_n)}{\delta^2_{n,b}} + 1+\sum_{k=K_{b-1}+1}^{K_b}\left((k-u_{a^*})^{-2}+(k-u_n)^{-2}\right)\nonumber\\
&\leq \frac{8\ln  (K_b-u_n)}{\delta^2_{n,b}} + 1+\sum_{k=1}^{\infty}2k^{-2} \nonumber\\
&\leq \frac{8\ln  (K_b-u_n)}{\delta^2_{n,b}} + 1+ \frac{\pi^2}{3}.
\end{align}

Since the sampling regret in epoch $b$ is $\text{SR}'_b=\sum_{n\neq a_b^{*}}\beta \delta_{n,b}\mathbb{E}[\theta_{n,b,K_b}] $, $\text{SR}'_b$ is $O(\ln K)$.

The handover regret
\begin{align}
\text{HR}'_b& =VC\sum_{n\in\mathcal{A}_{m,b}}\mathbb{E}[S_{m,n,b}]\nonumber\\
&\leq C\left(2\sum_{n\neq a_b^{*}} \mathbb{E} [\theta_{n,b,K_b}] + 1 \right) \nonumber\\
&\leq C \left(2\sum_{n\neq a_b^{*}}\left[\frac{8\ln  (K_b-u_n)}{\delta^2_{n,b}} + 1+ \frac{\pi^2}{3} \right]+1\right ).
\end{align}

Thus the handover regret $\text{HR}'_b$ is also $O(\ln K)$.

By summing  $\text{SR}'_b$ and $\text{HR}'_b$, and consider totally $B$ epochs, the total regret for each task is $O(B\ln K)$.

\end{document}